\documentclass{article}
\usepackage{amsmath,amsthm,amssymb,mathtools}
\usepackage{apxproof}
\usepackage{hyperref}
\usepackage{tikz}
\usepackage{pgfplots}
\usepgfplotslibrary{polar}
\usepackage[margin=1.5in]{geometry} %
\usetikzlibrary{shapes}
\usepackage{accents}
\usepackage{tikz}
\usepackage{pgfplots}
\usepgfplotslibrary{polar}
\usetikzlibrary{shapes}
\usepackage{nicefrac}
\usepackage[round]{natbib}

\newcommand{\R}{\mathbb{R}}
\newcommand{\E}{\mathbb{E}}
\renewcommand{\P}{\mathbb{P}}
\newcommand{\argmax}{\arg\,\max}
\newtheorem{theorem}{Theorem}
\newtheorem{claim}[theorem]{Claim}

\newtheorem{corollary}{Corollary}

\newtheorem{definition}{Definition}
\newtheorem{lemma}{Lemma}
\newtheorem{proposition}{Proposition}
\newtheorem{remark}{Remark}

\newtheorem{assumption}{Assumption}
\renewenvironment{proof}[1][Proof]{\textbf{#1.} }{\ \rule{0.5em}{0.5em}}
\theoremstyle{definition}

\begin{document}

\title{Opaque Contracts}

\author{Andreas Haupt and Zo\"{e} Hitzig\footnote{We thank especially Eric Maskin, Ben Golub and Shengwu Li. We also thank Benji Niswonger, Chang Liu, David Laibson, Ed Glaeser, Kathryn Holston, Jeremy Stein, Matthew Rabin, Mohammad Akbarpour, Shani Cohen, Yannai Gonczarowski and Zachary Wojtowicz.}}
\maketitle

\begin{abstract}
        Firms have access to abundant data on market participants. They use these data to target contracts to agents with specific characteristics, and describe these contracts in opaque terms. In response to such practices, recent proposed regulations aim to increase transparency, especially in digital markets. In order to understand when opacity arises in contracting and the potential effects of proposed regulations, we study a moral hazard model in which a risk-neutral principal faces a continuum of weakly risk-averse agents. The agents differ in an observable characteristic that affects the payoff of the principal. In a \emph{described contract}, the principal sorts the agents into groups, and to each group communicates a distribution of output-contingent payments. Within each group, the realized distribution of payments must be consistent with the communicated contract. A described contract is \emph{transparent} if the principal communicates the realized contract to the agent ex-ante, and otherwise it is \emph{opaque}. We provide a geometric characterization of the principal's optimal described contract as well as conditions under which the optimal described mechanism is transparent and opaque.
\end{abstract}

\section{Introduction}
Firms have access to abundant data on consumers and employees. They use these data to tailor contracts and market rules to specific participants, optimizing pricing, payment and content decisions with predictive algorithms. 

Such tailored rules are often described to market participants in opaque terms. Ride-hailing apps deploy complex pricing schemes and shape drivers' expectations about payment through summary statistics. Platforms personalize content to particular users and offer limited descriptions about how users' personal information influences the content they see. Employers have intricate rubrics for determining employee pay and promotions, but may describe these rubrics to employees and prospective employees incompletely. In response to such practices, recent laws and proposed regulations in the United States and the European Union aim to increase transparency in consumer and labor markets.\footnote{For instance the EU's \href{https://ec.europa.eu/commission/presscorner/detail/en/ip_22_2545}{Digital Services Act}, proposed in October 2022, states that platforms must offer ``clear, meaningful and uniform information about the parameters used" in a personalized market. Meanwhile, in the US, the \href{https://www.congress.gov/117/bills/s2024/BILLS-117s2024is.pdf}{Filter Bubble Transparency Act}, proposed in June 2021, suggested that any ``platform that uses an opaque algorithm" must ``provide notice to users of the platform that the platform uses an opaque algorithm that makes inferences based on user specific data." \href{https://www.washingtonpost.com/business/2023/01/25/pay-transparency-laws-companies-salary-disclosure/}{New laws in California, New York}, and a handful of other states require employers to post salary ranges for all advertised positions.}

When do firms offer opaque contracts, and how does opacity affect consumers and employees? In this paper, we study how firms (and other principals) design and describe contracts for agents. In a \emph{described contract}, a principal communicates different contracts to different groups of agents. Within each group, the ex-post distribution of outcomes must be consistent with the (possibly stochastic) contract communicated ex-ante. 

     The optimal described contract resolves a novel tradeoff. Relative to transparent contracts, opaque contracts benefit the principal by creating a wedge between the agents' expected payments and true payments. At the same time, opacity hurts the principal when agents are risk averse, by introducing new uncertainty into the contract. The principal must compensate the agents for taking on additional risk. 
     
     We begin with a stylized example that introduces the key elements of our model.

 \subsection{Introductory Example}\label{sec:intro_example_body}

A firm plans to introduce year-end bonuses for employees to induce higher effort. Since they cannot contract on effort, the firm will offer bonuses to employees who meet a performance target. The firm commits to the bonus schemes they communicate ex-ante---honoring commitments is important for employee trust and retention, and helps the firm steer clear of legal troubles. 

The firm has a wealth of data on employee performance that allows them to estimate the effective cost of paying each employee a bonus of \$x. These data reveal trends about employees in two divisions of equal size: employees in Division 0 are more expensive to incentivize---inclusive of taxes and other administrative and compliance costs, it costs four times more per dollar to incentivize employees in Division 0 than it does to incentivize employees in Division 1. 

Suppose the employees are risk-averse in payments $(x)$ and have quadratic effort ($a$) costs. In particular, assume that agents have utility $u(a,x) = a x^{\nicefrac{1}{2}} - \nicefrac12 a^2$, where $a$ is effort and the probability that the agent meets their performance target. Suppose further that the firm is risk-neutral and the value to the firm of each employee meeting their performance target is the same. Then the firm's optimal division-wide bonuses scheme are given by: 

\begin{center}

\fbox{\parbox{3.85in}{\textsc{Division 0}\medskip

{\small Employees who meet their performance targets receive bonuses of $\nicefrac13$.

\vspace{2mm}

\hfill ($\hat x_{0}$)}}}

\vspace{2mm}

\fbox{\parbox{3.85in}{ \textsc{Division 1}\medskip

{\small Employees who meet their performance targets receive bonuses of $\nicefrac43$.

\vspace{2mm}

\hfill ($\hat x_{1}$)}}}
\end{center}
that is, $x_0 = \frac13$ and $x_1=\nicefrac43.$ Given these bonues, the employees' optimal actions are $a_0 = \nicefrac{1}{\sqrt{3}}$ and $a_1 = \nicefrac{2}{\sqrt{3}}$.\footnote{Where $a_0$ and $a_1$ are determined by the first-order condition with respect to $a$ of the utility function of employees in Division 0 and Division 1, respectively.} Assuming that the value of each employee meeting their target is normalized to 1, the firm's payoff is $v = \nicefrac12 a_0(1 -  x_0) + \nicefrac12 a_1(1 -  \nicefrac{1}{4}x_1) = \nicefrac{1}{\sqrt{3}} \approx .55.$ 

Now suppose the firm instead decides to make a firm-wide commitment, offering the same bonus scheme to all employees. Taking advantage of the fact that the firm is both \textit{designing the bonus scheme} and \textit{choosing what to say about it}, they could take a more opaque approach. They could say: 

\begin{center}

\fbox{\parbox{3.85in}{\textsc{All Employees}\medskip

{\small Employees who meet their performance targets receive an average bonus of $\nicefrac{9}{8}$. Half of employees will receive of a bonus of $\nicefrac14$ and the other half will receive a bonus of 2.

\hfill ($\hat x_{01}$)}}}
\end{center}
and in fact carry out the scheme
$$ x_{01} = \begin{cases} \nicefrac14 & \text{if in Division 0} \\ 2 & \text{if in Division 1}.\end{cases} $$
Under this scheme, all employees take the same action  assuming that the contract they face is equally likely to be $\nicefrac14$ or 2, i.e $a = \nicefrac12 (\nicefrac14^{\nicefrac12}) + \nicefrac12 ( 2^{\nicefrac12})$.\footnote{Where $a$ is determined by taking the first-order condition of all employees' utility function with respect to $a$.} The principal's payoff is  $v = \nicefrac12 a(1 -  x_{01}(0)) + \nicefrac12 a\left(1 -  \nicefrac{1}{4}x_{01}(1)\right) = \nicefrac{5}{8}\left(\nicefrac14 + \nicefrac{1}{\sqrt 2}\right) \approx .6.$  Thus, the firm can do better by introducing a firm-wide, instead of division-by-division bonus scheme. This is because in the opaque scheme, the firm induces the same action from all employees, but then pays the employees differently. Compared to the transparent scheme, employees in Division 0 take a higher action in the opaque scheme but receive lower pay, while employees in Division 1 take a lower action but receive higher pay. On net, this benefits the firm because they can get higher actions overall but pay more to the ``cheap" employees in Division 1 than to the ``expensive" ones in Division 0. 
 
 Note that the firm's gain from the opaque scheme requires that the employees cannot deduce that their payments in fact depend on their division. Although some employees may try to infer how the firm decides who gets the lower bonus versus the higher bonus, there is an overwhelming number of factors to consider---it could depend on division, but it could also depend on region, on whether the employee's position is client-facing, on whether the employee works remotely, and so forth. In the face of such complexity, a modest modeling approach is to assume that employees can do is take the firm's communication at face value.

So far we have only considered two possibilities: in the first ``transparent" scheme, the firm communicates a different contract to each division, and in the second ``opaque" scheme, the firm communicates a single contract to all employees. Could the firm get an even higher payoff from the bonus scheme by additionally restructuring the divisions? That is, suppose the firm moved some of the employees in Division 1 to Division 0, creating Divisions $0'$  (mix of ``expensive" and ``cheap" employees) and Division $1'$ (all ``cheap" employees). Would the firm get an even higher payoff from offering an opaque scheme within the newly created divisions? And why stop there---could the firm do better by creating any arbitrary grouping of cheap and expensive agents, and communicating different opaque contracts to these groups? In this example, the answer is no: The firm cannot do strictly better than offering a single opaque contract to the whole firm.\footnote{The firm could also split the employees into divisions or any other grouping that preserved the composition of the overall firm, and attain the optimal payoff. That is, the firm could create two divisions Division 0$'$ and 1$'$ in which $\frac12$ of the employees in each division came from Division 0 and the other half from Division 1. The calculation for the exact optimal described contract can be found in \autoref{sec:intro_example}---it is in fact not too different from the contract $\hat x_{01}$.} The model in this paper develops tools for understanding the firm's optimal strategy in the scenario presented here.

\subsection{Overview}
The preceding example illustrates the key components of our moral hazard model with descriptions. A principal faces a continuum of weakly risk-averse agents, who each have a payoff-relevant observable characteristic. The principal wants to incentivize effort, but effort is costly for the agents and non-contractible. So, the principal offers contracts contingent on an output that is correlated with effort.

The principal chooses a \textit{described contract}, which has three elements. First, the principal chooses a \textit{sorting function} that sorts agents into different contracts based on their observable characteristics. In the introductory example, the sorting function determines whether all employees are grouped together or sorted into divisions, as well as the composition of ``cheap" and ``expensive" agents in each division. Second, the principal chooses a contract to \textit{communicate} to each group of agents. In the example, these are the statements $(\hat x_0, \hat x_1)$ in the transparent contract and $\hat x_{01}$ in the opaque contract. Finally, the principal chooses a contract to \textit{realize} for each specific agent, with the \textit{consistency} requirement that within each group, the distribution of realized payments is consistent with the communicated contract. In the transparent scheme in the example, the realized contract  $(x_0, x_1)$ is trivially consistent with the communicated contract $(\hat x_0, \hat x_1)$ because the two coincide exactly. In the opaque scheme in the example, the realized contract $x_{01}$ is consistent with the communicated contract $\hat x_{01}$ because it results in the communicated distribution of payments.

The principal knows that agents take the communicated contract at face value, and solves for the optimal described mechanism. We focus on when the optimal described contract is \textit{transparent}---that is, when the realized contract is communicated to the agent ex-ante. In such cases, it would be redundant for regulators to impose transparency requirements, as firms are already incentivized to be transparent. When a described contract is not transparent, we say it is \textit{opaque}. A particular type of opaque contract is a \textit{fully coarse} contract: a described contract is fully coarse if there is only one communicated contract. When the optimal described contract is fully coarse, the principal may have an easier time implementing it---for example, if the firm in the introductory example did not have the power to additionally restructure the divisions, it would have only the optimal transparent and the optimal fully coarse contract to choose from.\footnote{In \autoref{sec:coarse}, we study a modified version of the problem in which the principal is constrained to offering \textit{coarse} contracts, i.e. contracts that can group agents together based on their characteristics, but all agents with a particular characteristic must belong to the same group.}

In \autoref{sec:state_ind}, we provide a geometric characterization of the principal's problem. Our approach relies on concavification techniques familiar from \citet{aumann1995repeated} and the literature on persuasion and information design \citep[ff]{kamenica2011bayesian}. This technique allows us to gain insight into a joint contract design and information design problem, where the two parts of the design problem interact non-trivially: The optimal output-contingent payment scheme (contract design) depends on the composition of the group of agents to whom it is communicated (information design), and vice versa. This geometric characterization allows us to investigate the \textit{value of opacity }for the principal, defined as the different between the principal's optimal opaque mechanism and the principal's optimal transparent mechanism. In addition to investigating the value of opacity for the principal, we study how opacity affects agent welfare.

After using the geometric characterizations to derive sufficient conditions that are useful for applications, we present the result that captures the key tradeoff of the paper. As agents' risk-aversion increases, the value of opacity converges to zero. Opacity helps the principal by introducing a wedge between agents' expected payments and realized payments, which may be profitable. But opacity also introduces uncertainty---when agents are risk-averse, this uncertainty must be compensated with higher payments.\footnote{For a particularly simple illustration of this tradeoff, see \autoref{sec:intro_example}, which compares the example in \autoref{sec:intro_example_body} to a nearly-identical example in which employees are risk-neutral.} 

In the canonical single-agent moral hazard problem \citep{holmstrom1979moral, grossmanhart}, the principal faces a tradeoff between incentives and insurance: increasing the spread between payments for different outputs  motivates the agent to exert effort, but also discourages the agent by introducing additional risk into the contract. Our setting differs from a canonical contracting environment only in that there is a continuum of agents who each have an observable characteristic, and thus with these two ingredients the principal has an opportunity to offer an opaque contract. Opacity introduces a new element into this tradeoff. When comparing the potential benefits from moving from a transparent contract to an opaque contract, the principal trades off three quantities: the \textit{incentive increase} from \textit{some} agents who take a higher action, the \textit{incentive decrease} from \textit{some} agents who take a lower action, and the \textit{risk decrease} from \textit{all} agents' risk aversion. Our result shows that in the limit, the decrease from risk outweighs the net benefits from incentives.

Next, in \autoref{sec:application}, we apply our results to understand the design and communication of per-ride payment schemes for drivers on ride-hailing platforms. Here, the principal is a ride-hailing platform (e.g. Uber, Grab, DiDi) and the agents are drivers. The drivers' observable characteristic is whether they are available to drive during high-demand or low-demand periods. The drivers exert costly hidden effort: they must drive some distance from their home on the outskirts of the city into town---driving further into town increases the chances that the driver will pick up a ride. The firm chooses per-ride payment schemes and decides whether to communicate transparently (articulating how demand influences per-ride payments) or opaquely. This application allows us to show in detail how different parameters of the model influence the platform's optimal strategy. Further, it illustrates that while opacity on average increases driver welfare relative to transparency, it is not Pareto improving in general---a surprising finding with subtle policy implications.

Before concluding the paper, we discuss the related literature in \autoref{sec:lit}. This discussion includes an interpretation of our model and results through the lens of the ``Bayesian persuasion" paradigm. A brief conclusion in \autoref{sec:conclusion} suggests directions for future work.  
    
\section{Model}\label{sec:model}

There is a continuum of agents, each of whom takes a costly action $a \in A$ that affects the payoff of a principal. The principal does not observe the action, but observes output $q \in  Q$ that is informative about the action $a$, and distributed according to $\pi:  A \to \Delta (Q)$. We denote the probability of a particular output $q$ by $\pi(q \mid a)$, with $\sum_{q \in Q} \pi(q\mid a)=1$.

Each agent contracts with the principal in a particular state of the world $s \in S$, which is observed by the principal and may or may not be observed by the agents depending on the application. The state $s$ is distributed in the population according to distribution $f$, and the state space $S$ is finite. The agents maximize expected utility, and have symmetric von Neumann-Morgenstern utility functions $u \colon   A \times S \times X \to \R$, where $X$ is a set of feasible outcomes. The principal is risk neutral with an objective function $v \colon  A \times  S \times  X \to \R$. 

 The principal chooses a \textit{described contract}, which has three features. The first feature of a described contract is a collection of \textit{communicated contracts} $\hat g_k\colon Q \to \Delta(X)$ where $K$ is a set of contract labels and each $\hat g_k$ is a contract. We denote the implied probability distribution over outcomes given communicated contract $\hat g_k$ by $\hat P_{kq} \coloneqq \P(\hat g_k(q))$.  The principal also chooses an assignment rule $\sigma \colon S \to \Delta(K)$ which assigns agents to communicated contracts based on the state in which they arrive. We denote the  distribution of communicated contracts given state $s$ by $\mu_{s}\coloneqq \mathbb P(\,\cdot \mid s)$. Finally, the principal chooses a \textit{realized contract} $g_k: Q \times S \to \Delta(X)$ for each $k \in K$ which specifies the outcome for an agent who receives contract $k$ in state $s$. We denote the \textit{induced distribution of outcomes} given realized contract $g_k(q,s)$ by $P_{kqs}\coloneqq \P(g_k(q,s)).$ Note that the domain of a realized contract $g_k$ is $Q \times S$ while the domain of a communicated contract $\hat g_k$ is $Q.$

To each agent, the principal communicates only a single contract $\hat g_k$. Ex-post, the agent observes outputs and outcomes for all agents who received the same contract. That is, an agent who receives a contract labelled $k$ sees statistics on outputs and outcomes for all agents who received the contract labeled $k$---they see a ``database" containing $\hat P_{kq}$.\footnote{An alternative and equivalent interpretation is to view this as a legal constraint that is common knowledge: The principal will be punished by some third-party if she lies, and the agent knows this, and the Principal knows that the agent knows this... and so forth.}

Since the agents who receive contract $\hat g_k$ have access to the database $\hat P_{kq}$, the principal must choose realized contracts that are \textit{consistent} with the communicated contracts. In order to define consistency, we first introduce notation for the distribution of outcomes for a contract labelled $k$ for output $q$. We call this quantity the \textit{observed distribution of outcomes.}

\begin{definition}
    Given realized contract $g_k(q,s),$ the \emph{observed distribution of outcomes} is a probability distribution 
    $$P_{kq} \coloneqq \sum_{s\in S}P_{kqs} \frac{\mu_s(k) f(s)}{\sum_{s' \in S}  \mu_{s'}(k)f(s')}.$$
\end{definition}
A realized contract is \textit{consistent} with a communicated contract if the observed distribution of outcomes from the realized contract ($P_{kq}$) is the same as the distribution of outcomes implied by the communicated contract.
\begin{definition}[Consistency]
    A realized contract $g_k$ is consistent with a communicated contract $\hat g_k$ if $\hat P_{kq} = P_{kq}$
   for all $q \in Q$.
\end{definition}
In sum, the principal chooses a \textit{described contract} which contains a collection of consistent communicated contracts and realized contracts, as well as a sorting function. 

\begin{definition}[Described contract] 
A \emph{described contract} $((\hat g_k)_{k \in K}, ( g_k)_{k \in K}, \sigma)$ consists of communicated contracts $\hat g_k(q)$, realized contracts $g_k(q,s)$ and an assignment rule $\sigma \colon S \to \Delta(K),$ where for each $k\in K$, $g_k(q)$ is consistent with $\hat g_k(q,s).$
 \end{definition}
The following timeline summarizes and clarifies the timing of the game: 
 \begin{enumerate}
     \item The principal chooses a described contract $(( g_k)_{k \in K}, (\hat g_k)_{k \in K}, \sigma).$
     \item To each agent $i$ in state $s$, 
     \begin{enumerate}
         \item The principal communicates contract $\hat g_k$ with probability $\mu_s(k)$.
          \item The agent accepts or rejects the contract.

          \begin{itemize}
          \item If the agent rejects, she gets reservation utility 0. 
              \item If the agent accepts, she chooses action $a_k^*$, assuming the mechanism is $\hat g_k$. 
          \end{itemize}
     \item The outcome $x$ is realized according to $g_k(q,s)$.
     \end{enumerate}
    \item The principal and agents' utilities are realized. Agents who received contract $k$ observe $P_{kq}.$
 \end{enumerate}
The principal's program can be written as:\footnote{In summation form, the principal's expected utility is:$$ \sum_{s \in S}
 \sum_{k \in K}v(a_k^*,g_k(q(a_k^*),s),s) \mu_s(k) f(s).$$} 
\begin{equation}
    \max_{(\hat g_k)_{k \in K},(g_k)_{k \in K}, \sigma} \E_{s\sim f,\, k \sim \mu_s}[v(a_k^*, g_k(q(a_k^*),s),s)]
\end{equation}
\begin{center}
    subject to
\end{center}
\begin{align*}
\hat P_{kq} &= P_{kq}
 \text{ for all }k \in K, q \in Q \tag{Consistency}\\
    a_k^* &\in \argmax_a \E[ u(a, \hat g_k(q(a)), s)]\,\text{ for all } k\in  K \tag{IC} \\
    0 &\leq \E[ u(a_k^*, \hat g_k(q(a_k^*)), s)]\, \text{ for all } k\in  K. \tag{IR}
\end{align*}

We say that a described contract that solves this program is an \textit{optimal described contract}. There are cases when the principal's optimal described contract is not unique. In such cases, we assume that the principal chooses the optimal described contract that maximizes agent welfare. This assumption simplifies the exposition and is necessary for characterizing agent welfare in \autoref{sec:agent_welfare}. Our analysis will focus on properties of the optimal described contract. To summarize, the optimal described contract, together with the agents' actions $(a_k^*)$ constitute a (welfare-maximizing) subgame perfect equilibrium. 

Note that in the final stage of the game (step 3), both the principal's and agents' realized utility depend on the set of \textit{realized} outcomes, determined by $g_k(q,s)$. This feature of the model implies that when the agents make their decision in step 2b of the game, they may not know as much about their outcome as the principal knows. We say that a described contract is \textit{transparent} if the realized contract is communicated to the agent ex-ante. This occurs when, in the described contract, there is a one-to-one correspondence between the agent's state and the label of the contract to which the agent is assigned. Or, put another way, the described contract is transparent if for each agent the distribution of outcomes implied by the communicated contract is the same as the outcomes in the realized contract.

\begin{definition}[Transparent described contracts]
A described contract $((\hat g_k)_{k \in K}, ( g_k)_{k \in K}, \sigma)$ is \emph{transparent} if $\sigma \colon S \to \Delta(K)$ is a bijective function.
\end{definition}
We contrast transparent described mechanisms with \textit{opaque} described mechanisms. 

\begin{definition}[Opaque described contract]
A described contract $((\hat g_k)_{k \in K}, ( g_k)_{k \in K}, \sigma)$ is \emph{opaque} if it is not transparent.
\end{definition}

In an opaque described mechanism, the principal takes advantage of the fact that they face a population of agents, and can thus describe the mechanism with varying amounts of detail without ``lying." The principal is constrained by consistency, but otherwise has the freedom to, for example, describe the realized contract for a particular agent stochastically even if it is in fact deterministic. A particular opaque described mechanism will be of interest in both our theoretical analysis and applications---when does the principal communicate the same mechanism to all agents? We say that a described mechanism is \textit{fully coarse} when the principal communicates the same contract to all agents. 

\begin{definition}[Fully coarse described contracts]
A described contract $((\hat g_k)_{k \in K}, ( g_k)_{k \in K}, \sigma)$ is \emph{fully coarse} if $|K|=1$.
\end{definition}

In \autoref{sec:coarse}, we will study the principal's optimal described contract when the principal must respect an additional \textit{no-arbitrary-differential-treatment} constraint. This constraint says that for any two agents $i, j$ who arrive in the same state $s$, the principal must communicate the same mechanism $\hat g_k,$ i.e. $\sigma(s)$ is injective.\footnote{The \textit{no-aribitrary-differential-treatment} constraint may arise in applications for feasibility or fairness reasons. For example, in the introductory example in \autoref{sec:intro_example_body}, if the Division 0 and Division 1 were in fact different regional offices in different jurisdictions, it may be infeasible to communicate contracts to groups composed of employees from both regions. Or, if the difference underlying the differences in employee cost were in fact due to different job descriptions in the two divisions, it may be perceived as unfair for employees with the same job descriptions to receive different contracts.}

\section{Properties of the Optimal Described Contract}\label{sec:state_ind}

The principal's problem can be understood as a joint contract design and information design problem, where the two parts of the problem interact non-trivially. Nonetheless, we can characterize the principal and agent's utility at the optimal described contract through a concavification argument similar to that introduced in \citet{aumann1995repeated} and adapted to the study of ``persuasion" in \citet{kamenica2011bayesian}.

In this section, we restrict our attention to cases in which the agent's utility does not depend on the state. 

\begin{assumption}
The agent's utility is state-independent, i.e. $u(a,s,x) = u(a,x)$ for all $s \in S, a \in A, x \in A$.
\end{assumption}

\subsection{Geometric Characterizations}\label{sec:value_function}
 We define the principal's \textit{value function} $V: \Delta(S) \to \R$ to be the value of the objective function when the principal can communicate only one contract to all agents (i.e. the principal is constrained to \textit{fully coarse} described contract). Note that when the principal is choosing among fully coarse contracts, the range of $\sigma$ has a single element, so there is only a single realized contract $g_k$ and a single communicated contract $\hat g_k$, which we will call $g$ and $\hat g$, respectively. That is, 
 when considering only fully coarse contract, the only realized contract is $g(q,s) \coloneqq g_k(q,s),$ and the only communicated contract is $\hat g(q) \coloneqq \hat g_k(q)$.

\begin{definition}[Principal value function]\label{def:p_value_function}
   The \emph{principal's value function} is given by
\begin{equation}\label{eq:value_function}
    V(f) \coloneqq \max \left\{ \E_{s\sim f}[v(a^*, g(q(a^*),s), s) ] \, \middle\vert \begin{split}
    g &\colon Q \times S \to \Delta(X) \\ \, a^* &\in \argmax_a \E[u(a, \hat g(q(a)))] \\ 0 &\leq  \E[u(a^*, \hat g(q(a^*)))] \end{split} \right\}. 
\end{equation}
\end{definition}
We next introduce the concave closure of the principal's value function, which is pictured in the graph on the left in \autoref{fig:illustration_of_closures}. The figure presents the case where the state is binary, and we can thus identify a distribution $f$ with the probability of one of the states. So, a particular probability distribution $f \in \Delta(S)$ is represented by a point on the x-axis.

\begin{definition}[Closure of principal value function]
 The \emph{concave closure of the principal's value function} is \begin{equation}
    \overline{V}(f) \coloneqq \sup\{z \mid (f, z) \in \operatorname{co}(V)\}\label{eq:closure}
\end{equation}
where $\operatorname{co}(V)$ is the convex hull of the graph of the value function $V$.   
\end{definition}
In our first characterization result, we show that the concave closure $\overline{V}(f)$ of the principal's value function at population distribution $f$ is the principal's utility at the optimal described contract. The logic behind this result is similar to the logic in models of ``persuasion" \citep{kamenica2011bayesian}, but important differences lurk beneath the surface in both analysis and interpretation---we note two differences here and return to the discussion in \autoref{sec:lit}.

The first key difference is that in persuasion models, the concave closure of the value function represents the principal's (or ``sender's") utility from the optimal choice of a single object (``the statistical experiment") whereas here the the concave closure of the value function represents the principal's utility from the optimal choice of two objects (the composition of groups and the output-contingent payment scheme). Here the value function represents the principal's utility from what can be thought of as an information policy (the sorting of agents into groups) \textit{and} a contract choice. It is not obvious how the contract choice interacts with the optimal informational scheme. 

Note, second, what this implies about the principal's commitment: in persuasion models, the principal (sender) must commit to a statistical experiment---an assumption which may fit a limited range of settings, since the adherence to a statistical experiment is often difficult or impossible to monitor. Meanwhile, in our model of contracting with descriptions, the consistency requirement amounts to a standard form of commitment: The principal commits only to a (possibly stochastic) communicated contract, which is as easy to monitor as any commitment to a (possibly stochastic) contract in canonical contract theory.\footnote{In this sense our results connect to \citet{lin2022credible}. We return to the topic of persuasion, opacity and commitment in \autoref{sec:lit}, and there discuss an interpretation of our model and results through the lens of persuasion.}

\begin{theorem}\label{thm:value_func_characterization}
Given a distribution $f\in \Delta(S)$, the principal's utility at the optimal described contract lies on the concave closure of the principal's value function, at $\overline V(f)$. 
\end{theorem}

\begin{proof}
Consider an arbitrary point $V^* = \overline{V}(f)$ on the concave closure of the value function. We prove two statements about this point: (i) There exists a described contract $(g_k, \hat g_k, \sigma)$ that achieves $V^*,$ and; (ii) there is no described contract that achieves higher utility for the principal than $(g_k, \hat g_k, \sigma)$ at $f$.

We begin with (i). If $\overline V(f) = V(f)$ then the optimal fully coarse described contract achieves $V^*$, by definition of $V(f)$. If $\overline V(f) \neq V(f)$, then there is a maximal convex subset $I \subseteq \Delta(S)$ such that $f \in I$ and $V(f)\neq \overline V(f).$  We denote the boundary of the set $I$ by $\partial I$. By Carath\'{e}odory's theorem, there is a set of at most $\lvert S \rvert $ boundary points, $\{ \partial_1, \partial_2, \dots, \partial_{\lvert S \rvert} \} \subseteq \partial I$ such that
\begin{align*}
\overline V(f) &= \sum_{k = 1}^{|S|}\lambda_k V(\partial_k) &f &= \sum_{k = 1}^{|S|}\lambda_k \partial_k.
\end{align*}
We may interpret every boundary point $\partial_k \in \Delta(S)$ as a distribution of states $s$ given that the contract is $k$.

We define a sorting function $\sigma \colon S \to \Delta(K)$ such that the distribution of contracts given state $s$ ($\mu_s$) satisfies\footnote{To see that this is a probability distribution, observe: $\sum_{k \in K} \mu_s(k) = \sum_{k \in K}\frac{\lambda_k\partial_k(s)}{f(s)} =\frac{\sum_{k \in K}\lambda_k\partial_k(s)}{f(s)} =  \frac{f(s)}{f(s)} =1.$})
\[
\mu_s(k) = \frac{\lambda_k\partial_k(s)}{f(s)}.
\]

Consider for each distribution $\partial_k$ the fully coarse communicated contract $\hat g_k$ and denote the (unique) realized contract for the fully coarse contract consistent with $\hat g_k$ by $g_k$. We claim that $(\sigma, (\hat g_k)_{k\in K}, (g_k)_{k\in K})$ defines a described contract that achieves principal utility of $\overline{V}(f)$. That this is a described contract, i.e., that $\hat g_k$ is consistent with $g_k$, is a direct consequence of the fact that $\hat g_k$ and $g_k$ are consistent as the communicated resp. described contract of a fully coarse contract. 

Next, we show that this described contract gives the principal utility of $\overline V(f)$. Note that contracts ($\hat g_k$, $g_k$) yield utility $V(\partial_k)$ (by definition of $V$). By the principal's risk neutrality, we can express the principal's utility contract-by-contract,
\begin{align*}
 \sum_{k \in K} \sum_{s \in S}f(s) \mu_s(k) V(\partial_k) &=\sum_{k \in K} \sum_{s \in S} f(s)\frac{\lambda_k\partial_k(s)}{f(s)} V(\partial_k) \\
&= \sum_{k \in K}\lambda_k V(\partial_k) \sum_{s \in S}\partial_k(s) = \sum_{k \in K}\lambda_k V(\partial_k) = \overline{V}(f).
\end{align*}

Next we prove statement (ii). Consider an arbitrary described contract $((\hat g_k)_{k\in K} , (g_k)_{k\in K}, \sigma)$ that yields principal utility $V \in \R$. Introduce the probability distribution on $S$ that describes the composition of agents to whom contract $\hat g_k$ is communicated: 
\[
\rho_k(s) \coloneqq \frac{ f(s) \mu_s(k)}{\sum_{s \in s} f(s) \mu_s(k)}.
\]
Observe that the utility that the principal derives from the mass of agents $\sum_{s \in S} f(s) \mu_s(k)$ that receive the communicated contract $\hat g_k$ may not be larger than the principal utility from the optimal fully coarse contract, $V(\rho_k)$ for this group. This is by definition---each $\hat g_k$ is a ``coarse" contract in the sense that a single contract is communicated to the mass of agents who receive $k$. So the \textit{optimal} coarse contract for the mass of agents $\sum_{s \in S} f(s) \mu_s(k)$ necessarily delivers weakly higher utility to the principal than the arbitrary contract $\hat g_k$. Hence, $V \le \sum_{k \in K} \sum_{s \in S} f(s) \mu_s(k) V(\rho_k)$. Note that $(\rho_k, V(\rho_k))$ is in the graph of $V$, and 
$$\sum_{k \in K} \sum_{s \in S} f(s) \mu_s(k) =\sum_{s \in S}  f(s) \left(\sum_{k \in K} \mu_s(k) \right) = \sum_{s \in S}  f(s) = 1,$$ 
so $f(s)\mu_s(k)$ are weights in a convex combination. Hence, $V \le \overline{V}(f)$ by the definition of the concave closure.
\end{proof}
Theorem \ref{thm:value_func_characterization} is powerful because it will typically be easier to solve for the principal's optimal fully coarse contract than it will be to solve for the principal's optimal described contract. This is because solving for the optimal described contract involves a nested optimization, optimizing over optimal contracts for all possible partitions of agents into groups. Meanwhile, the optimal coarse contract is the optimal contract for a single partition of agents into groups. 

We next define another object that will give further insight into the optimal described contract without requiring the an optimization over optimal contracts for all partitions of agents. The \textit{extremal closure of the principal's value function} is the supremum of the convex hull of the graph of $V$ evaluated only at the ``extreme" distributions in which the entire population belongs to a single state. 
\begin{definition}[Extremal closure of principal value function]
 The \emph{extremal closure of the principal's value function} is \begin{equation}
    V^T(f) \coloneqq \sup\{z \mid (f, z) \in \operatorname{co}(V(f') \colon f' \in I)\},
\end{equation}
   where $I=\{\delta_s \colon s \in S\} \subset \Delta(S)$ is the set of point mass distributions at all $s \in S.$ 
\end{definition}
On the left in \autoref{fig:illustration_of_closures}, we see that when the state is binary, the closure of the extremal principal value function is a line connecting $V(0)$ and $V(1).$ The extremal closure of the principal's value function gives the principal's utility at the optimal transparent contract.
\begin{proposition}\label{prop:transparent}
    Given a distribution $f \in \Delta(S)$, the principal's utility at the optimal transparent contract is $V^T(f),$ i.e. the extremal closure of the principal's value function evaluated at $f$.
\end{proposition}

\begin{proof}
   We proceed in two steps. We consider an arbitrary point $V^T(f)$, and show: (i) that it is attained by a transparent contract, and (ii) that there is no transparent contract that achieves higher utility. 

   We begin with (i). A point $V^T(f)$ can be decomposed as $V^T(f) = \sum_{s \in S} \lambda_s V(\delta_s)$
   by the definition of $V^T$, where each $\delta_s$ is the point mass distribution at $s$. As the convex combination is extremal, it must be that $\lambda_s = f(s)$. Each $V(\delta_s)$ is the principal's utility at the optimal fully coarse contract $\hat g_s$ given that the population distribution is $\delta_s$, by the definition of the the principal's value function. So consider a described contract with each $\hat g_s$ as the fully coarse contract at $\delta_s,$ and each $g_s$ the (unique) realized contract that is consistent with $\hat g_s.$ Define $\sigma \colon S \to \Delta(K)$ such that 
   \begin{equation}\label{eq:transp_sigma}
          \mu_s(k) = \begin{cases} 1 & \text{ if }s=k \\ 0 & \text{ otherwise.}\end{cases} 
   \end{equation}
   Note that this $\sigma$ is bijective. The described contract with communicated contract $\hat g_s$ as the fully coarse contract at $\delta_s$ and $\sigma \colon S \to \Delta(K)$ the bijective sorting function that satisfies \eqref{eq:transp_sigma} (i.e. maps each state $s$ to contract $\hat g_s$) yields value $V^T(f)$ 
   \begin{align*}
   \sum_{k \in K} \sum_{s \in S} f(s) \mu_s(k) V(\delta_k) & = \sum_{s \in S} f(s)\sum_{k \in K}  \mu_s(k) V(\delta_k) \\
   & = \sum_{s \in S} f(s)\sum_{k \in K}  \mu_s(k) V(\delta_k) = \sum_{s \in S} f(s) V(\delta_s) =  \sum_{s \in S} \lambda_s V(f_s) .
   \end{align*}

   Next, we show (ii)---that there is no transparent contract that achieves higher utility at $f$ than $V^T(f).$ Consider a transparent contract with $((\hat g_s)_{s \in S}, (g_s)_{s \in S}, \sigma)$. As this contract is transparent, for each $s \in S$, there exists a unique $k \in K$ such that $\mu_s(k) = 1$. Note that this leads to distributions $\mu_k(s)$ that are concentrated on a single $s \in S$. It must be the case that each $\hat g_s$ and $g_s$ correspond to optimal coarse contracts on these degenerate distributions. Hence, the value from each contract is bounded from above by $\sum_{s \in S} f(s) \mu_s(k) V(\delta_{k})$, where again $\delta_k$ denotes a point-mass distribution at at $k$. In sum, the principal value from all groups is
   \[
   \sum_{k \in K} \sum_{s \in S} f(s) \mu_s(k) V(\delta_{k}) = \sum_{s \in S} f(s) \mu_s(s) V(\delta_{s}) \le  V^T(f),
   \]
   where the bound is a result of the convex combination being from points on the graph of $V$.
\end{proof}
This proposition is valuable because, together with Theorem \ref{thm:value_func_characterization}, it allows us to gain insight into the optimal described contract, without computing the optimal described contract directly. The optimal transparent contract is much easier to solve for than the optimal described contract: it requires only solving for the optimal contract for $|S|$ groups---it does not requiring searching through different possible combinations of groups. 

To recap: we have defined three objects. The principal's value function, the closure of the principal's value function, and the extremal closure of the principal's value function. These three objects are shown on the left in \autoref{fig:illustration_of_closures}, and give the principal's utility at the optimal fully coarse contract (Definition \ref{def:p_value_function}), the optimal described contract (Theorem \ref{thm:value_func_characterization}), and the optimal transparent contract (Proposition \ref{prop:transparent}), respectively. 

The characterization results thus far give us tools that will be helpful in understanding the principal's side of the problem---they will allow us to understand the shape of the optimal described contract and how much the principal benefits from its ability to offer opaque contracts. Given our motivation to understand regulatory problems, we also need tools to understand the agent's side of the problem---how does the principal's optimal described mechanism affect agent welfare?

So, next we define three objects concerning agent utility that are analogous to the three objects just defined for the principal. The \emph{agent's value function} is the agents' welfare at the principal's optimal fully coarse described mechanism, where \textit{agent welfare} for actions $a_k$, realized contract $g_k$ and sorting function $\sigma(s)$, is the average of the agents' ex-post utility, i.e.
$$ \sum_{s \in S} \sum_{k \in K} u(a_k, g_k(q,s))\mu_s(k)f(s).$$

\begin{definition}[Agent value function]
    The \emph{agent's value function} is given by 
\begin{equation}\label{eq:agent_value_function}
    U(f) \coloneqq \max \left\{ \E_{s\sim f}[u(a^*,g(q(a^*),s), s) ]  \, \middle\vert \begin{split}
    g &\colon Q \times S \to \Delta(X) \\\, a^* &\in \argmax_a \E[u(a, \hat g(q(a)))] \\ 0 &\leq  \E[u(a^*, \hat g(q(a^*)))]    \end{split} \right\}.
\end{equation}
\end{definition}

\begin{figure}[h!]
    \centering
    \begin{tikzpicture}[scale=0.7]
\begin{axis}[
legend style={at={(.8,.45)},anchor=north,draw=none, font=\Large},
    axis lines=middle,
    xlabel=$f \in \Delta(S)$,
    xlabel style={at={(axis description cs:.6,-0.055)},anchor=north},
    ylabel style={at={(axis description cs:-0.1,.75)},anchor=south},
    xmin=0, xmax=1.1,
    ymin=0, ymax=1.25,
    samples=100,
     xtick={.29,1},
    xticklabels={$s^*$,1},
    ytick=\empty,
     label style={font=\Large}
    ]

    \addplot[domain=0:1,blue,smooth, line width=3pt] {.75-(.75-1.5*x)^3};
     \addlegendentry{$V(f)$}
\addplot[domain=0:.29,dashed,gray,line width=3pt] {.78-(.75-1.5*x)^3};
\addlegendentry{$\overline{V}(f)$}
    \addplot[domain=0:1, black, dotted, line width = 2pt] {.3281+.8509*x};
    \addlegendentry{$V^T(f)$}
\addplot[domain=.3:1,dashed,gray,line width=3pt] { 0.640563*x+0.56};

    \addplot+[only marks, mark=*, mark size=4, mark options={fill=black}, color=black] coordinates {(0,.3281) (1, 1.1826)};
    \addplot+[only marks, mark=square*, mark size=4, mark options={fill=black}, color=black] coordinates {(.29,.7256)};
\end{axis}
\node[text=black, align=center, above] at (current axis.north) {Principal Value};
\end{tikzpicture} \hspace{5mm}\begin{tikzpicture}[scale=0.7]
\begin{axis}[
legend style={at={(.8,.45)},anchor=north,draw=none,font=\Large},
    axis lines=middle,
    xlabel=$f \in \Delta(S)$,
    xlabel style={at={(axis description cs:.6,-0.055)},anchor=north},
    ylabel style={at={(axis description cs:-0.05,.75)},anchor=south,rotate=90},
    xmin=0, xmax=1.1,
    ymin=0, ymax=1.25,
    samples=100,
     xtick={.29,1},
    xticklabels={$s^*$,1},
    ytick=\empty,
     label style={font=\Large}
    ]

    \addplot[domain=0:1,blue,smooth, line width=3pt] {.95-(.8-x)^2};
     \addlegendentry{$U(f)$}
\addplot[domain=0:.29,dashed,gray,line width=3pt] {.99-(.8-x)^2};
\addlegendentry{$\tilde{U}(f)$}
    \addplot[domain=0:1, black, dotted, line width = 2pt] {.31+.6*x};
    \addlegendentry{$U^T(f)$}
\addplot[domain=.3:1,dashed,gray,line width=3pt] {0.3*x+0.62};

    \addplot+[only marks, mark=*, mark size=4, mark options={fill=black}, color=black] coordinates {(0,.31) (1, .91)};
    \addplot+[only marks, mark=square*, mark size=4, mark options={fill=black}, color=black] coordinates {(.29,.7)};
\end{axis}
\node[text=black, align=center, above] at (current axis.north) {Agent Welfare};
\end{tikzpicture}
    \caption{Principal value and agent welfare at optimal described, fully coarse, and transparent contract, $|S|=2$.}
\label{fig:illustration_of_closures}
\end{figure}
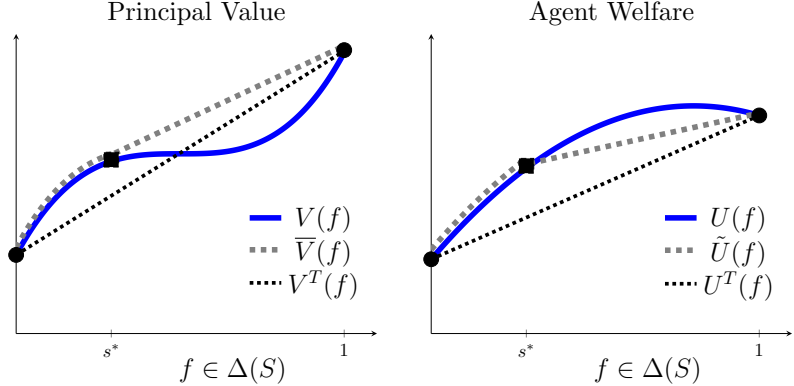

The agent value function is the agents' welfare under the principal's optimal choice of a coarsely described contract. We will next define the \emph{implied agent value function}. 

\begin{definition}[Agent implied value function]
Denote $\mathcal I$ the set of (inclusion-)maximal convex subsets $I \subseteq \Delta(S)$ such that $V(f) \neq \overline V(f)$. The agents' \emph{implied value function} is
\[
\tilde U (f) =\begin{cases}
    U(f) & \forall I \in \mathcal{I} : f \notin I \\ 
    \max \{ z \mid (f, z) \in \operatorname{co} (U|_{\partial I})\}&  \exists I \in \mathcal{I} \colon f \in I,
\end{cases}
\]
where $\partial I$ denotes the boundary of the set $I$.
\end{definition}
An example of the agent's implied value function is shown on the right in \autoref{fig:illustration_of_closures}, alongside the principal's value function. Recall, the figure presents the case where the state is binary, and we identify a distribution $f$ with the probability of one of the states. There are two important intervals in the closure of the principal's value function: on the interval $f \in [s^*, 1],$ the principal's value function is equal to its closure ($V(f) = \overline V(f)$);  on the interval $f \in [0,s^*]$ the principal's value function is not equal to its closure. So, in this case, there is a single maximal subset $I$ such that $V(f) \neq \overline V(f)$, and it is $I = [0,s^*]$ (and so $\mathcal I = \{I\}$).

These intervals define the agents' implied value function: on the interval $f \in [s^*,1]$, i.e. $f \not\in I$, the agents' implied value function is given by the agents' value function $U(f)$; on the interval $f \in [0,s^*]$, i.e. $f \in I$, the agent's implied value function is given by the convex hull of the graph of $U$ evaluated at the boundary of set $I.$

As we did for the principal's objective function, we next define the \textit{extremal closure of the agents' value function}. In the binary state case pictured in \autoref{fig:illustration_of_closures}, the extremal closure is a line connecting the points $U(0)$ and $U(1).$ In general, the extremal closure of the agents' value function is obtained by taking the supremum of the convex hull of the graph of the implied value function defined only on points $\delta_s$, i.e. points that represent a point-mass distribution at some $s \in S.$

\begin{definition}[Extremal closure of agent value function]
 The \emph{extremal closure of the agents' value function} is \begin{equation}
    U^T(f) \coloneqq \sup\{z \mid (f, z) \in \operatorname{co}(U(f') \colon f' \in I)\},
\end{equation}
where $I=\{\delta_s \colon s \in S\} \subset \Delta(S)$ is the set of point mass distributions $\delta_s$ at all $s \in S.$
\end{definition}
Analogous to Theorem \ref{thm:value_func_characterization} and Proposition \ref{prop:transparent}, the next proposition establishes that the agents' implied value function at $f$ is agents' welfare at the principal's optimal described contract, and that the extremal closure of the agents' value function at $f$ is agents' welfare at the principal's optimal transparent contract.
\begin{proposition}\label{prop:agent_characterization}
    Given a distribution $f\in \Delta(S)$, agent welfare at the optimal described mechanism is given by $\tilde U(f)$, and agent welfare at the optimal transparent mechanism is given by $U^T(f)$.
\end{proposition}

As the proofs for the two statements in Proposition \ref{prop:agent_characterization} are similar to the corresponding proofs for the principal, they can be found in \autoref{sec:proofs}.  

Note that there are many statistics of the distribution ex-post agent utility that may be of interest in applications. While welfare (i.e. expected ex-post agent utility) is perhaps the most natural one, other statistics such as the variance of the distribution may also lend important insight into how opacity affects equity. It is worth noting that our geometric characterization method works precisely because we are interested only in agent welfare, which is linear in individual agent's utility. Our characterizations in Proposition \ref{prop:agent_characterization} rely on the linearity of agent welfare in the same way that our characterization in Theorem \ref{thm:value_func_characterization} and Proposition \ref{prop:transparent} rely on the principal's risk neutrality.

\subsection{When is Opacity Optimal?}

We have stressed that solving for the optimal described contract may be difficult when the state space $S$ is large. In applications, it may suffice to know whether the optimal described contract has particular properties---such as whether the optimal described contract is opaque or transparent. Our geometric characterizations lead directly to a series of sufficient conditions that serve as tests for particular properties of the optimal described contract. We begin with a corollary of the characterization of the principal's utility at the optimal described contract (Theorem \ref{thm:value_func_characterization}). 
\begin{corollary}\label{cor:suff_cond_coarse}
If $V$ is globally weakly concave, i.e. $\frac{\partial^2 V}{\partial f^2} \leq 0$ for all $f \in \Delta(S)$, then there is an optimal described contract that is fully coarse.
\end{corollary}
Given the characterization in Theorem \ref{thm:value_func_characterization}, we know that in value function is equal to its concave closure if and only if there exists an optimal described contract that is fully coarse. This directly gives us Corollary \ref{cor:suff_cond_coarse}, a sufficient condition for an optimal described contract to be fully coarse, which depends only on the value function $V(f),$ and not its closure.

Note that if there is an optimal described contract that is fully coarse, the set of optimal contracts is infinite. The principal can choose a $\sigma$ that creates any partition of agents in which the composition of agent characteristics within each partition cell is the same as the distribution $f$. This is not the case when the optimal described contract is transparent. If there is an optimal described contract that is transparent, then it is the uniquely optimal described contract. This is because there is only one sorting function $\sigma$ that creates the partition that corresponds to the optimal transparent contract---the partition $\{\{s\} \colon s \in \}.$

Similarly, we can use the geometric characterization of the principal's optimal transparent contract to get a sufficient condition for when the principal's optimal described mechanism is transparent. 
\begin{corollary}\label{cor:suff_cond_trans}
The optimal described contract is transparent if $V(f)$ is globally weakly convex, i.e. $\frac{\partial^2 V}{\partial f^2} \geq 0$ for all $f \in \Delta(S)$.
\end{corollary}

A direct implication of Corollaries \ref{cor:suff_cond_coarse} and \ref{cor:suff_cond_trans} is that we can conclude whether the optimal described contract is coarse from understanding the concavity of the principal's coarse value function in only a region of $\Delta(S)$, and same for transparency. 
\begin{corollary}
If $V$ is convex in a neighborhood of some $f^* \in \Delta(S)$, then there is no optimal described contract that is fully coarse. If $V$ is concave in a neighborhood of some $f^* \in \Delta(S)$, then there is no optimal described contract that is transparent. 
\end{corollary}

\subsection{The Value of Opacity}\label{sec:agent_welfare}

In addition to generating sufficient conditions that can be useful for determining whether the optimal described contract is transparent or fully coarse, we can use our geometric characterizations to understand how much the principal gains through the opportunity to be opaque. To motivate this exercise, note that the optimal transparent contract is a collection of contracts that solve a textbook principal-agent model with moral hazard. So, the difference between the optimal described contract and the optimal transparent contract gives us insight into how principals facing many agents with different observable payoff-relevant characteristics can benefit from jointly designing contracts and groups of agents to whom they are communicated. This can help to predict in which settings contracts are most likely to be opaque, and can illustrate, in settings where it may be costly to carry out an opaque scheme, whether opacity is worth it to the principal. 

    Formally, the \emph{value of opacity} at $f \in \Delta(S)$ is the difference between the optimal described contract at $f$ and the optimal transparent contract at $f$, i.e. $\overline{V}(f) - V^T(f)$.

    We show that, in a wide class of problems, as agents become more risk averse, the value of opacity diminishes. In order to demonstrate this, we restrict attention to a particular class of utility functions studied in the contracting literature. 

\begin{assumption}\label{ass:sep_utility}
    Agent utility $u(a, x)$ can be written $h(a)\tilde{u}(x) - c(a)$ where $\tilde{u}$ is real-valued, continuous, increasing, and concave, $c(a)$ is continuous, and $k(a)$ is continuous and strictly positive. 
\end{assumption}
Under these assumptions, the realized contract in any optimal described mechanism is deterministic.  

\begin{lemma}\label{lem:deterministic}
   Assume Assumption \ref{ass:sep_utility} holds. Then in any optimal described contract $((\hat g_k)_{k \in K},$ $( g_k)_{k \in K}, \sigma)$, for each $k \in K$, the realized contract $g_k(q,s)$ is deterministic, i.e. $g_k\colon Q \times S \to X.$ It follows that in any transparent contract, for every $k \in K$, the communicated $\hat g_k \colon$ is also deterministic. 
\end{lemma}
This lemma follows directly from results in \citet{holmstrom1979moral} and \citet{grossmanhart} for the single agent moral hazard problem. Their arguments proceed by showing that any stochastic contract is dominated by a deterministic one. This fact holds because a principal in the canonical moral hazard problem faces a trade-off between providing incentives and providing insurance---giving higher payments for observable outputs creates higher incentives, but also increases the spread between outcomes in the agent's lottery. Under Assumption \ref{ass:sep_utility}, a stochastic contract, i.e. a contract that offers a lottery conditional on a particular output, increases risk for the agent without improving incentives. The same argument applies for realized contracts $g_k(q,s)$ in our setting.

Although, the tradeoff that the principal in our model faces in choosing a \textit{realized} contract is exactly the same as in the single agent setting, the principal's choice of a \textit{described contract} introduces a different kind of tradeoff  between incentives and insurance. When the principal faces a continuum of agents, and has the power to describe contracts, if opacity is valuable, then it also improves incentives \textit{on average} while introducing risk \textit{for all agents}. The value of opacity thus trades off three quantities: the \textit{incentive gains} from some agents taking a higher action than in the optimal transparent contract, the \textit{incentive losses} from some agents taking a lower action than in the optimal transparent contract, and the \textit{insurance losses} from introducing more uncertainty into the communicated contract (i.e. more uncertainty from the agent's perspective).

In the next Proposition, we consider a constant absolute risk aversion utility function with risk aversion parameter $\rho$, $\tilde u(x) = 1-e^{-\rho x}$ . In reference to a principal's value function with respect to $u_\rho$, we will use a subscript ${V}_\rho$.

\begin{proposition}
    Assume Assumption \ref{ass:sep_utility} holds. Then as agents become more risk averse, the value of opacity $\overline V_\rho(f) - V^T_\rho(f)$ converges to zero, i.e. 
    $$\lim_{\rho \to \infty} \overline{V}_\rho(f)- V^T_\rho(f) = 0.$$
\end{proposition}
\begin{proof}
Let $f \in \Delta(S)$ be a type distribution. Let $(\hat g, g)$ be an optimal fully coarse contract. We show that there is is a transparent contract $\tilde g_s$ whose principal utility approaches the one of $(\hat g, g)$.

Define the transparent contract $\tilde g_s(q)$ as the realized contract $g(q,s)$ with an additional payment for each $q$ and $s$ that makes agents indifferent between the lottery $\hat g_s(q)$ and $g_s(q,s)$. It must be that $g_s(q,s)$ is deterministic by Lemma \ref{lem:deterministic}. Note that under these contracts, agents choose the same action as under the coarse contract $\hat g$. The amount of transfer that the principal needs to pay for the agent satisfies
$\E[\tilde u_\rho (\hat g(q(a)))] = \E[\tilde u_\rho ( \tilde g(q(a)) + x)]$. Note that the only randomness in the second expectation is with respect to $q$. We show that for any $x > 0$ there is $R$ such that for all $\rho > R$, 
\begin{equation}
\E[\tilde u_\rho (\tilde g(q(a)))] < \E[\tilde u_\rho (\tilde g(q(a))+x) ].\label{eq:eventuallysmaller}
\end{equation}
This means that the amount of extra payment needed to make the agent indifferent converges to zero. We prove a more general statement about lotteries, which implies \eqref{eq:eventuallysmaller} for lotteries $l = \hat g (q(a))$ and $l_i = g_s(q,s)$ 
\begin{claim}
For every compound lottery $l = \sum_{i=1}^m p_i \cdot l_i$, every $i=1, 2, \dots, m$, and every $x > 0$, there is $R$ such that for all $\rho \ge R$, $\E[\tilde u_\rho (l)] < \E[\tilde u_\rho (l_i + x) ]$, where $l_i + x$ means the addition of x with certainty to all elements of $l_i$.
\end{claim}
This result means that a sufficiently risk averse agent will accept the worst part of a lottery plus any fixed positive amount. As we can unravel compound lotteries, we may, without loss, assume $l_i$ to be deterministic and $l$ be a (non-compound) lottery. We then have
\begin{align*}
1-\exp(-\rho(l_i + x)) > \sum_{j=1}^n p_i (1-\exp(-\rho l_j))& \iff \exp(-\rho(l_i + x)) < \sum_{j=1}^n p_i \exp(-\rho l_j)\\
&\iff
1 < \sum_{j=1}^n p_i \exp(-\rho (l_j - l_i - x) ).
\end{align*}
Note that by assumption, at least one of the terms $l_j - l_i - x$ must be negative (namely the one for $j=i$), which means that for high enough $\rho$, the corresponding summand $p_i\exp(\rho x)$ will be larger than $1$. As all other summands are nonnegative, this proves the claim. 
\end{proof}

Given the focus on transparency in regulatory conversations about algorithms, it is also valuable to understand how agent welfare under the optimal described contract (which may be opaque) compares to the optimal transparent contract. We next consider the \textit{welfare increase from opacity}, which is the difference between agent welfare at the optimal described contract for distribution $f \in \Delta(S)$ and the optimal transparent contract at distribution $f$, i.e. $\tilde U(f) - U^T(f).$ We present the following sufficient condition, which is a corollary to Proposition \ref{prop:agent_characterization}.
\begin{corollary}
If $U(f)$ is globally weakly concave, then the welfare increase from opacity is weakly positive. If $U(f)$ is globally weakly convex, then the welfare increase from opacity is agent welfare is weakly negative.
\end{corollary}

\section{Application: Design and Communication of Ride-Hailing Payments}\label{sec:application}

We next turn to an application: a ride-sharing platform (principal) choosing how to jointly design and describe payment schemes to drivers (agents). 

We focus on this setting because ride-hailing platforms use many data to set payment schemes, and tailor their communications to different drivers. One commentator summarizes  ride-hailing platforms' communication practices with drivers as follows: ``using what they know about drivers, their control over the interface and the terms of transaction, they channel the behavior of the driver in the direction they want it to go."\footnote{\url{https://www.nytimes.com/interactive/2017/04/02/technology/uber-drivers-psychological-tricks.html}}

\subsection{Setting}
We consider drivers located in the outskirts of the city, who exert effort $a$ to drive into the city to increase their chance of picking up a passenger. Without loss, we let $a$ be the probability that the driver picks up a passenger, given that she exerts effort $a$. The platform cannot contract on the distance driven into the city ($a$) because it is difficult to verify. The binary contractible output $q \in \{0,1\}$ that is correlated with the drivers' effort is whether the driver picks up a passenger $(q=1)$. The platform's output-contingent contract is thus a per-ride payment for the driver, $g$. 

We assume that drivers are risk averse in money (with square root utility), and that effort $a$ is quadratic. The drivers' utility is given by $u(a,g) = a \sqrt{g}- \frac12 a^2.$

The platform has plentiful data on the drivers. A particular characteristic $s$ of the drivers that is of interest to the platform is whether drivers tend to drive during high-demand ($s=h$) periods or low-demand $(s=\ell)$ periods. The low state occurs with probability $f(\ell) = \alpha$. 

The platform is risk-neutral and the state affects their utility function in the following way. When a driver completes a ride, the platform earns $b_s$ and pays the driver $\tau_s g$, where $b_s$ and $\tau_s$ are constants that depend on the state. So, the platform's objective in state $s$ is given by $v(a,g,s) = a(b_s - \tau_s g),$ where $b_s,\tau_s >0$ for all $s.$ 

The ``productivity" parameter $b_s$ can be understood as the commission that the platform takes in state $s$. For instance, if $b_\ell = 1$ and $b_h>1,$ it would imply that the platform can take a higher commission during high-demand periods---this may be the case if the platform charges surge fares to passengers without proportionally increasing driver pay.\footnote{Although it's hard to find clear publicly available information about ride-hailing platforms' pricing strategies, it is widely documented that the platform drivers' compensation is not increased in proportion to surging passenger fares. See, e.g. \url{https://www.washingtonpost.com/technology/2021/06/09/uber-lyft-drivers-price-hike/}.} Meanwhile, the ``effective payment" parameter $\tau_s$ can be understood as the effective cost of paying the driver \$$g$ per ride. For example, if $\tau_\ell >\tau_h,$ it would imply that paying a driver \$$g$ per ride costs more in the low-demand state than it costs to pay a driver \$$g$ per ride in the high-demand state. This may be the case if the platform is operating in a jurisdiction where they must compensate drivers for idling time, as is the case in New York City.\footnote{\url{https://www.theverge.com/2019/2/1/18206737/nyc-driver-wage-law-uber-lyft-via-juno}} If the platform must pay drivers for idling time, then, when demand is low ($s=\ell$), the platform pays more per ride than just the per-ride payment \$$g$.  

The platform chooses payment schemes and messaging tactics about their payment schemes to incentivize (especially inexperienced) drivers to drive into the city. A communicated contract is a communication that summarizes the payment scheme, perhaps sent as a reminder to incentivize the driver to get on the road. For example, the platform could send a message such as ``Drivers in your area are currently earning \$10 per ride on average." Or, it might be more specific and send different messages to the two types of drivers: to the $s=h$ drivers, it could say ``Since you tend to drive when demand is high, you will earn \$15 per ride." 

In addition to choosing the communicated contract, the platform can strategically sort into different cohorts who receive the same message. The platform can stagger its' notifications so that different groups of agents get different messages at different times. For legal reasons, the platform cannot lie---their communicated payment rules must be \textit{consistent} with the realized payment rules within each group.

From a regulatory perspective, it is valuable to know when the platform offers a transparent described contract. In such settings, the minimal legal requirement of ``consistency" is strong enough to make the platform tell agents everything there is to know about their payment schemes---imposing transparency requirements would be redundant.

\subsection{Discussion}

Suppose that it costs the same amount to pay drivers in both states and when demand is high $(s=h)$, drivers are more productive $(b_h>b_\ell)$. The difference in value of the drivers' effort comes from the platform's ability to charge surge pricing to passengers---the platform can charge passengers more for rides when there is high demand.
\begin{remark}
 Assume $\tau_\ell=\tau_h$ and $b_\ell<b_h$. Then the optimal described contract is transparent. 
\end{remark}

\begin{figure}[h!]
    \centering
    \includegraphics[width=.45\textwidth]{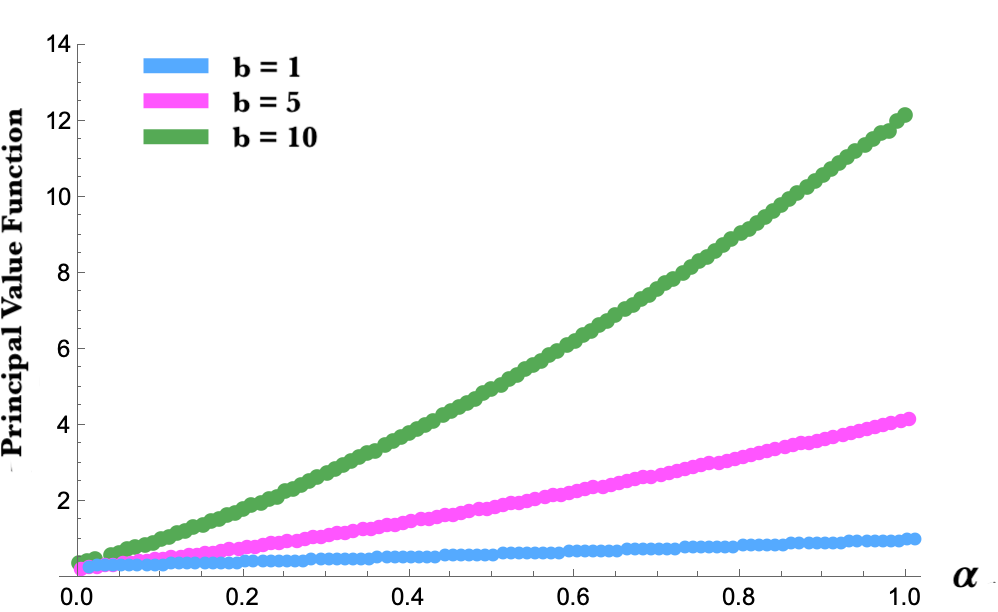}   \includegraphics[width=.45\textwidth]{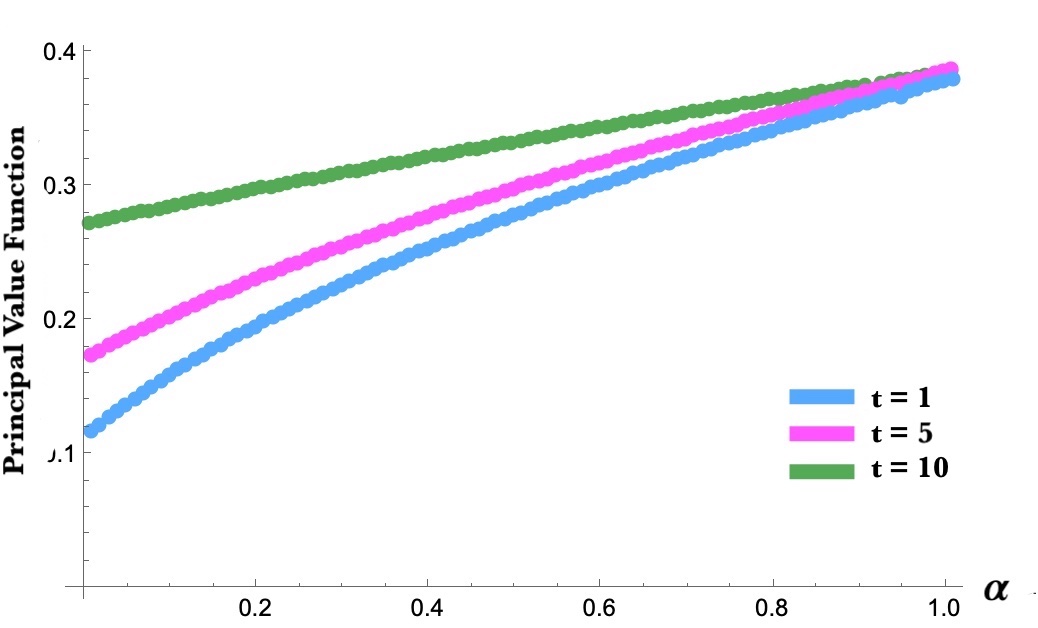}
    \caption{Principal Value Function: $\tau=1, b \in \{1,5,10\}$ (left); $b=1, \tau \in \{1,5,10\}$ (right) }
    \label{fig:uber_ex}
\end{figure}
The platform's optimal described contract is transparent coarse in this case, and \autoref{fig:uber_ex} shows that the principal's value function is globally weakly convex for $b_h \in \{1, 5, 10\}$. In this case, the platform would say specifically to drivers who drive in high demand periods, ``If you drive when demand is high, your compensation will be \$X." And to drivers who drive in low demand periods, they would say ``If you drive when demand is low, your compensation will be \$Y." This is because the platform wants to incentivize drivers to get on the road when demand is high, and there's no benefit to introducing extra uncertainty into the contract. 

Next, suppose that the only difference between drivers is the effective cost of paying them.
\begin{remark}
Assume $b_\ell = b_h$ and $\tau_\ell > \tau_h.$ Then the optimal described mechanism is fully coarse.
\end{remark}
In this case, the platform's optimal strategy is to give an opaque description of the payments, and not mention how payments will in fact depend on demand. This makes drivers in the low-demand state think that there is some chance that they get higher payments, and so the principal can induce higher actions from them without paying the extra per-dollar cost. As \autoref{fig:uber_ex} shows, the principal's value function is globally concave for $\tau \in\{1,5,10\}. $

\section{Related Literature}\label{sec:lit}

 This paper relates most closely to the literature on mechanism design with an informed principal, and models of joint information design and mechanism design. 
 
 In the canonical informed principal problem, the principal has private information that the agent doesn't have. The agent takes the principal's announcement of a contract as a signal about the principal's private information, and the agent solves for a Bayes-Nash equilibrium of the signaling game. Though originally introduced in screening settings \citep{myerson1983mechanism,maskin1990principal,maskin1992principal}, the informed principal problem has also been studied in the presence of moral hazard \citep{beaudry1994informed, inderst2001incentive,mekonnen2021informed,clark2021informed}.
    
   Our set up is similar to such models in that the principal has information that the agent doesn't have---here the principal's private information is about the contract itself. But our model also contrasts with the informed principal problem in that the agent takes the principal's announcement ``at face value," subject to a consistency condition. That is, the agent does not take the contract as a ``signal" about the principal's private information. In some settings of interest, the agent may not even know that she is less informed than the principal, and we connect then to models of contracting with unawareness \citep{filiz2012incorporating,auster2013asymmetric}. In other settings that motivate our investigation, there are too many factors that could influence the agent's belief about what the contract is---in the example of \autoref{sec:application}, a ride-hailing platform driver and knows that the platform uses its extensive data to compute pricing schemes but cannot articulate the platform's type space or form a belief about it. Similar to models of add-on pricing \citep{ellison2005model} and shrouded attributes \citep{gabaix2006shrouded}, if the principal doesn't mention a relevant feature of the agents' environment, the agent simply won't think about it.  The model is thus similar in spirit to cursed equilibrium. In cursed equilibrium, players in a game fail to update their beliefs based on the information content of the other player's action  \citep{eyster2005cursed,esponda2008behavioral}. In our model, the agent does not condition on the ``information content" of the principal's choice of description---she does not solve for the equilibrium in a signaling game. 

The literature on persuasion and information design also assumes that agents (or ``receivers") do not think strategically about the principal's (or ``sender's") strategy---they don't have to because the principal commits to a signal structure about the state (see \citet{bergemann2019information} for a review). The agent receives a signal about an unknown state, updates her belief about the state given her knowledge of the signal structure to which the principal has committed, and takes an action. Although our analysis relies on the same techniques used in this literature, we differ in both substance and interpretation.

First, in pure information design, the principal changes the agent's utility from taking a particular action by changing her beliefs, rather than changing how she is paid. In our model, the principal jointly designs payments and the ``information" that agents have. Although there is an emerging literature developing models that join information design and mechanism design \citep{dworczak2020mechanism,boleslavsky2018bayesian,bergemann2022,kwak2022optimal}, these papers concern a principal who can design agents' information about an exogenous feature of the environment. For example, \citet{kwak2022optimal} in particular studies a moral hazard problem in which an exogenous unknown state affects output alongside the agent's action. The principal can design an experiment on this unknown state, controlling what the agent believes about her environment. By contrast, in our model, the principal controls what the agent believes about the contract itself.

Further, our model differs from information design models in interpretation. In our model, the agent does not know that the principal has committed to a signal structure and does not update her belief according to Bayes' rule taking the signal structure into account. Instead, in our model, all that a particular agent takes into account when making her decision is the (possibly stochastic) contract $\hat g_k(q)$ that is communicated to her. The fact that the principal's realized contracts $\hat g_k(q,s)$ have to be consistent with the communicated contracts disciplines the problem so that the principal has some flexibiliity in description, but is not in an ``anything-goes" environment. 

Despite these differences, it should not be too surprising that the analyses we carry out here can be entirely understood through the lens of persuasion. Our problem is in fact equivalent to a ``contracting with persuasion" problem in which the agent has a prior on an unknown state $s,$ and the principal commits to a contract $\psi(m,q,s)$ and a signal structure $\phi: M \to \Delta(S).$ The agent sees a signal realization $m \sim \phi(s)$ and then forms a posterior belief about the contract she faces, taking into account the principal's commitment. This is not the model we present for two reasons: it is does not well capture agents' often limited understanding of their environments in the settings of interest to us, and it relies on the assumption that principals can commit to signal structures. As such, the persuasion model to which ours most closely relates is \citet{lin2022credible}, which derives conditions under which it is without loss to replace the commitment assumption with a ``credibility" requirement. Their credibility requirement resembles our consistency requirement (though without payments): a disclosure policy is credible if the principal cannot profit from tampering with her signals while keeping the signal distribution unchanged.

\section{Conclusion}\label{sec:conclusion}
This paper introduced \textit{described contracts} into a moral hazard framework in which a principal contracts with a continuum of agents. We provided a geometric characterization of the optimal described contract, and conditions under which it is transparent, and conversely, opaque. We showed how opacity alters the canonical contracting tradeoff between incentives and insurance: opacity may create an (on average) profitable wedge between agents' expected and realized payments. 

In future work, we hope to better understand the effects of opacity on agents, as regulations that call for more transparency draw on arguments about how transparency affects agents like consumers  and workers. In this paper, we looked only at agent welfare---but opaque contracts lead to greater heterogeneity in ex-post agent outcomes than transparent mechanisms. So, if a regulator cares about ``equal treatment" of agents who have different characteristics, then allowing for described mechanisms is detrimental to the cause (compared to transparent or uniform contracts). In addition, opaque contracts, although they may lead to higher average welfare, also may entail violations of agents' \textit{ex-post} individual rationality. And even if agents are not led into a violation of their ex-post rationality constraint, they are ``misled" in the sense that their ex-post utility under the described mechanism is different from their ex-post utility \textit{had they known everything that the principal knew} about the intended payment.


\bibliographystyle{aer}
\bibliography{refs}

\appendix

\section{Proofs}\label{sec:proofs}

\textbf{Proof of Proposition \ref{prop:agent_characterization}.}
\begin{proof}
    Consider an arbitrary point $(f, U^*)$ where $U^* = \tilde U(f).$ If $\tilde U(f) = U(f)$ then $\tilde U(f)$ describes the average agent welfare at the optimal fully coarse contract, by definition. Since $\tilde U(f) = U(f)$ implies that $\overline V(f) = V(f),$ then at this point, the principal's optimal described mechanism is fully coarse (by Theorem \ref{thm:value_func_characterization}). 

    Next, consider points $f \in \Delta(S)$ for which $\tilde U(f) \neq U(f)$ and thus $V(f) \neq \overline V(f).$ Then there is a maximal convex subset $I \subseteq \Delta(S)$ containing point $f$ and points $f' \colon V(f') \neq \overline V(f')$. Denote the boundary of the set $I$ by $\partial I$. By Carath\'{e}odory's theorem, there is a set of at most $|S|$ boundary points $\{\partial_1, \partial_2, \dots, \partial_{|S|}\}$ such that
    $$\tilde U(f) = \sum_{k=1}^{|S|} \lambda_k U(\partial_k).$$
    Note that each boundary point $\partial_k$ is associated with a principal-optimal fully coarse contract for distribution $\partial_k$. Denote by $\hat g_k$ the optimal fully coarse contract at $\partial_k$, and by $g_k$ the (unique) realized contract consistent with $\hat g_k$. Now define $\sigma \colon S \to \Delta(K)$ such that $\mu_s(k)$ satisfies 
    $$ \mu_s(k) = \frac{\lambda_k \partial_k(s)}{f(s)}. $$
    Note that this choice of $\sigma$ is only one of the ones that the Principal is indifferent between. It yields, however, the maximal agent utility.

    We now show that the described contract $(\hat g_k)_{k \in K}, (g_k)_{k \in K}, \sigma)$ as defined gives the agents welfare of $\tilde U(f).$ Note that contracts $(\hat g_k, g_k)$ yield welfare $U(\partial_k)$ by the definition of $U$ (because these $(\hat g_k, g_k)$ are principal optimal fully coarse contracts). Agent welfare can thus be written:

    \begin{align*}
        \sum_{s \in S} \sum_{k \in K} u(a_k, g_k(q,s))\mu_s(k)f(s) & =   \sum_{s \in S} \sum_{k \in K} U(\partial_k)\mu_s(k)f(s)  \\
        &= \sum_{s \in S} \sum_{k \in K} U(\partial_k)\frac{\lambda_k \partial_k(s)}{f(s)}f(s)  \\
          &= \sum_{k \in K} U(\partial_k)\lambda_k  \sum_{s \in S}\partial_k(s) \\
             &=  \sum_{k \in K} U(\partial_k)\lambda_k  \\
             & = \tilde U(f)
    \end{align*}

By Theorem \ref{thm:value_func_characterization}, we know that $(\hat g_k)_{k \in K}, (g_k)_{k \in K}, \sigma)$ (where each $\hat g_k$ is the optimal fully coarse contract at $\partial_k$ and $\sigma$ is defined above) is the principal's optimal described contract.

\end{proof}

\section{Introductory Example: The Role of Risk Aversion}\label{sec:intro_example}

To contrast with the introductory example in \autoref{sec:intro_example}, we produce here the exact same example but with risk neutral agents. 

A firm that plans to introduce year-end bonuses for employees. The firm will offer a bonus schemes to employees who meet their performance targets, and commits to the schemes they communicate ex-ante (either because they have to commit, for legal reasons, or because they believe that honoring commitments is important for employee retention). 

The firm has a wealth of data on employee performance that allows the executive to estimate the effective cost of paying each employee a bonus of \$x. These data reveal trends about employees by division: employees in Division 0 are more expensive to incentivize---inclusive of taxes, administrative and compliance costs, and follow-on profits, it costs four times more per dollar to incentivize employees in Division 0 than it does to incentivize employees in Division 1. 

Suppose the employees are risk-neutral in money with utility functions $u(a,x) = a x - \frac12 a^2$, where $a$ is effort (and the probability that they meet their performance target). Then, assuming that: the divisions are of equal size, the firm is risk neutral, the value to the firm of each employee meeting their performance target is the same, the firm's optimal division-wide bonuses scheme would be to say

\begin{center}

\fbox{\parbox{3.85in}{\textsc{Division 0}\medskip

{\small Employees who meet their performance targets receive bonuses of $\nicefrac12$.

\vspace{2mm}

\hfill ($\hat x_{0}$)}}}

\vspace{2mm}

\fbox{\parbox{3.85in}{ \textsc{Division 1}\medskip

{\small Employees who meet their performance targets receive bonuses of $2$.

\vspace{2mm}

\hfill ($\hat x_{1}$)}}}
\end{center}
and to actually give bonuses $x_0 = \frac12$ and $x_1=2.$ Assuming that the value of each employee meeting their target is normalized to 1, the firms total payoff from the bonus program is given by $v = \frac12 a_0(1 -  x_0) + \frac12 a_1(1 -  \frac{1}{4}x_1) =\frac12 (x_0)^(\frac{1}{2})(1 -  x_0) + \frac12 (x_1)^(\frac{1}{2})(1 -  \frac{1}{4} x_1) =\nicefrac{5}{8}.$ The employees' optimal actions are determined by the first-order condition of their utility with respect to $a$.

Now suppose the firm decided to instead make a firm-wide commitment. Taking advantage of the fact that the firm is both \textit{designing the bonus scheme} and \textit{choosing what to say about it}, they could take a more opaque approach. They could say: 

\begin{center}

\fbox{\parbox{3.85in}{\textsc{All Employees}\medskip

{\small Employees who meet their performance targets receive an average bonus of $2$. 

\hfill ($\hat x_{01}$)}}}
\end{center}
and in fact carry out the scheme
$$ x_{01} = \begin{cases} 0 & \text{if in Division 0}\\ 4 & \text{if in Division 1}.\end{cases} $$
Under this scheme, all employees take the same action $(a=2)$ assuming that the contract they face is equally likely to be $0$ or 4. The principal's payoff is  $v = \nicefrac12 a(1 -  x_{01}(0)) + \nicefrac12 a\left(1 -  \nicefrac{1}{4}x_{01}(1)\right) = 1.$ 

 So, the firm can do better by introducing a firm-wide, instead of division-by-division bonus scheme. Note that this assumes that: (i) the employees believe the firm's commitment to the stochastic scheme, and (ii) the employees have no way of learning that in fact the Division 0 employees will get the lower bonuses while the Division 1 employees will get the higher bonuses. 

Notice that the benefit from introducing the optimal opaque scheme in this example, where agents are risk neutral, is greater than the benefit from introducing the optimal opaque scheme in the example in \autoref{sec:intro_example_body}, where the agents are risk averse. This is because the opaque scheme introduces uncertainty into the contract---this uncertainty benefits the principal because it allows the principal to create a wedge between the agents' expectations of their outcome and their actual outcome, but it also hurts the principal because the uncertainty leads risk averse agents to take a lower action. When the agents are risk neutral, this tradeoff disappears.

\section{No-arbitrary-differential-treatment}\label{sec:coarse}

In a described contract that satisfies no-arbitrary-differential treatment, if the principal communicates a different mechanism to agents $i$ and $j$ with $i\neq j$, it must be that the agents have different characteristics. The no-arbitrary-differential-treatment constraint can be thought of as a fairness requirement. The basic idea is that if two agents do not differ in $s$, their observable characteristic, then both agents should receive the same communicated contract $\hat g_k.$

\begin{definition}
A described contract $((\hat g_k)_{k \in K}, (g_k)_{k \in K}, \sigma)$ satisfies \emph{no-arbitrary-differential-treatment} if $\sigma$ is injective.
\end{definition}
Note that if $\sigma$ is injective, then $g_k(s,q)=g_k(q)$. That is, it is redundant to specify how the realized contract depends on both $k$ and $s$, since each $s$ maps to a single group $k$ by $\sigma$. 

We provide a geometric characterization of the principal's optimal described contract that satisfies no-arbitrary-differential-treatment.

\begin{definition}[Orthogonal concave closure]
The \emph{orthogonal concave closure} $\overline{V}^{\operatorname{orth}}$ of a function $V$ is given by 
$$ \overline{V}^{\operatorname{ort}}  \coloneqq \sup \left\{z \mid (F, z) \in \operatorname{co}^{\operatorname{ort}}(\operatorname{Graph}(V))\right \},$$
where $co^{\operatorname{ort}}(\operatorname{Graph}(V))$ is the orthogonal convex hull of the graph of $V$, i.e.
$$\operatorname{co}^{\operatorname{ort}}(\operatorname{Graph}(V)) \coloneqq \left\{\sum_{j} \lambda_j V(F_j) \colon \sum_j \lambda_j = 1, \lambda_j \geq 0, \lambda_i, \lambda_j > 0 \Rightarrow \langle F_i, F_j \rangle = 0\right\}. $$ 
\end{definition}
The orthogonal convex hull of the graph of $V$ is the hull of only those convex combinations of $V$ for which all vectors with a nonzero coefficient are orthogonal.

\begin{proposition}
For a given distribution $f \in \Delta(S)$, the value of the optimal described mechanism that satisfies no-arbitrary-differential-treatment lies on the orthogonal concave closure of the principal's value function.
\end{proposition}

\end{document}